\documentclass[conference]{IEEEtran}
\usepackage{amsfonts}
\usepackage{amssymb}
\usepackage{cite}
\usepackage[cmex10]{amsmath}
\usepackage{float}
\usepackage{color}
\usepackage{stfloats,fancyhdr}
\usepackage{amsmath}
\usepackage{amsthm}
\usepackage{algorithm}
\usepackage{algorithmic}
\usepackage{multirow}
\usepackage{changepage}
\usepackage[normalem]{ulem}

\newtheorem{lemma}{Lemma}

\newtheorem{remark}{Remark}

\IEEEoverridecommandlockouts

\ifCLASSINFOpdf
  \usepackage[pdftex]{graphicx}
  \DeclareGraphicsExtensions{.pdf,.jpeg,.png}
\else
  \usepackage[dvips]{graphicx}
  \DeclareGraphicsExtensions{.pdf}
\fi

\usepackage{subfigure}

\usepackage{fancybox,dashbox}

\begin{document}
%
\title{{Age-of-Information Dependent Random Access for Massive IoT Networks}
\thanks{The first two authors contributed equally to this work. The work of H. Chen is supported by the CUHK direct grant under the project code 4055126.}
}
\author{\IEEEauthorblockN{He Chen\textsuperscript{1}, Yifan Gu\textsuperscript{2}, and Soung-Chang Liew\textsuperscript{1}}%
\IEEEauthorblockA{\textsuperscript{1}Department of Information Engineering, The Chinese University of Hong Kong, Hong Kong SAR, China}
\IEEEauthorblockA{\textsuperscript{2}School of Electrical and Information Engineering, The University of Sydney, Sydney, Australia\\
\textsuperscript{1}\{he.chen, soung\}@ie.cuhk.edu.hk, \textsuperscript{2}yifan.gu@sydney.edu.au}}

\maketitle

\begin{abstract}
As the most well-known application of the Internet of Things (IoT), remote monitoring is now pervasive. In these monitoring applications, information usually has a higher value when it is fresher. A new metric, termed the age of information (AoI), has recently been proposed to quantify the information freshness in various IoT applications. This paper concentrates on the design and analysis of age-oriented random access for massive IoT networks. Specifically, we devise a new stationary threshold-based age-dependent random access (ADRA) protocol, in which each IoT device accesses the channel with a certain probability only when its instantaneous AoI exceeds a predetermined threshold. We manage to evaluate the average AoI of the proposed ADRA protocol mathematically by decoupling the tangled AoI evolution of multiple IoT devices and modelling the decoupled AoI evolution of each device as a Discrete-Time Markov Chain. Simulation results validate our theoretical analysis and affirm the superior age performance of the proposed ADRA protocol over the state-of-the-art age-oriented random access schemes.
\end{abstract}



\IEEEpeerreviewmaketitle

\section{Introduction}
Internet of Things (IoT) represents one of the most significant paradigm shifts recently, which can revolutionize the information technology and several aspects of everyday life such as living, e-health and driving; it envisions to transform every physical object into an intelligent individual that is capable of sensing, communicating and computing \cite{Lin2017IoTJ}. Ericsson foresaw that by 2021, there will be around 28 billion IoT devices and a large share of them will be empowered by wireless communication technologies \cite{Ericsson2016}. Analysts predicted that by 2025, the economic impact of the IoT could reach US\$11 trillion, or 11\% of global economic value, and by 2030 the IoT could influence nearly the entire economy \cite{McKinsey2015}.

A typical IoT network is made up of three main ingredients: 1) IoT devices, 2) communication network, and 3) information fusion nodes. The IoT devices are often deployed to observe a physical characteristic of the environment, e.g., temperature, pollution levels, or speed and location of a vehicle. The sensed data are transmitted through the communication network to the information fusion nodes where they are processed to extract meaningful information, e.g., control decisions or remote source reconstruction for predicting its information status evolution. Clearly, the accuracy of such output decisions, which affects the performance of various IoT-enabled applications, is heavily determined by the freshness of the data measurements of IoT devices at the information fusion nodes~\cite{Shreedhar2018Acp}.

Conventional performance metrics (e.g., throughput and delay) cannot adequately capture the information freshness. Specifically, due to random network delay, maximizing throughput or minimizing delay does not necessarily guarantee the freshest information to be observed at the receivers \cite{Sun2017tit}. In this context, the AoI concept was first introduced in \cite{Kaul2011mini} as a new metric to measure the information freshness at the destination side. AoI is a function of both how often packets are transmitted and how much delay packets experience in the system. The metric of AoI is of great importance in the IoT applications where the timeliness of information is crucial, and thus has attracted enormous attention recently, see e.g., \cite{kadota2018scheduling,wang2018skip,gu2019minimizing,maatouk2019minimizing,gu2019timely,wang2019minimizing,wang2020minimizing,wang2019minimizing2,li2020ageoriented} and references therein.

With the new metric of AoI, a fundamental design problem for large-scale wireless IoT networks is ``\emph{how to schedule the status updates of massive IoT devices to achieve a low network-wide AoI}". Though the analysis and optimization of AoI for various network setups have become an increasingly hot topic recently, there has only been limited work that attempted to answer the fundamental question given above  \cite{S.Kaul-Distributed-Centralized-MAC,R.Talak-Distributed-MAC,Kosta2019age,jiang2018timely,jiang2018decentralized,chen2019age}. Specifically, \cite{S.Kaul-Distributed-Centralized-MAC,R.Talak-Distributed-MAC,Kosta2019age} investigated age-independent stationary randomized policies, in which each transmitter sends its packet with some fixed probability that can be optimized ahead of time. These stationary randomized policies are easy to implement in a distributed manner; however, they have the shortcoming of not leveraging the instantaneous AoI information at the transmitter side. The work in \cite{jiang2018timely} designed a round-robin scheme for AoI minimization; however, such a scheme is incapable of dealing with the change of number of nodes in the network and thus may not be suitable for IIoT applications with nodes joining and leaving the system in a dynamic way (e.g., some nodes could switch to a sleep mode for saving energy). The follow-up work \cite{jiang2018decentralized} additionally assumed that nodes are provided with carrier sensing capabilities and proposed distributed schemes that have good performance in simulations; nevertheless, \cite{jiang2018decentralized} does not address how the parameters of the proposed algorithms should be designed.

Inspired by the aforementioned work, in this paper we aim to design and optimize an age-dependent stationary randomized policy for large-scale IoT networks which can be easily implemented in a decentralized manner. In our policy, the channel access probability (CAP) of each IoT node is predetermined as in \cite{S.Kaul-Distributed-Centralized-MAC,R.Talak-Distributed-MAC} but the CAP is age-dependent, in contrast to the age-independent counterpart schemes in \cite{S.Kaul-Distributed-Centralized-MAC,R.Talak-Distributed-MAC}. This paper makes two main contributions: (1) We devise a threshold-based distributed age-dependent random access (ADRA) protocol for massive IoT networks. Specifically, each device accesses the channel with a constant probability only when its instantaneous age exceeds a predetermined threshold; otherwise it will keep silent. (2) We develop an analytical framework for deriving a closed-form expression of the average AoI for each node in the network when the CAP and age threshold are given. Simulation results are provided to validate our analytical results and demonstrate the superiority of the proposed ADRA protocol over its conventional age-independent counterpart. {During the preparation of this paper, we noticed that a very recent work \cite{chen2019age} also proposed a similar random access policy. The active probability of each device in \cite{chen2019age} is based on conventional ALOHA backoff mechanisms, while we use a pre-determined CAP for each IoT device. Besides, the proposed random access policy in \cite{chen2019age} was optimized only for the case when the number of devices approaches infinity. At last, \cite{chen2019age} focused on the policy design and no closed-form performance analysis was provided.} Simulation results showed that our proposed scheme also outperforms the one in \cite{chen2019age}.


\section{System Model}
Consider an uplink IoT network consisting of an access point (AP) and $N$ IoT devices, denoted by $D_1,D_2,\cdots,D_N$, which aim to report their status as timely as possible to the AP via a common wireless channel.
The timeliness and freshness of the status updates from various IoT devices at the AP is quantified by the recently proposed AoI metric. As in \cite{S.Kaul-Distributed-Centralized-MAC}, time is divided into slots of equal durations and the transmission of each status update packet takes exactly one time slot. All IoT devices implement a slotted ALOHA-like random access protocol. Specifically, during each time slot, each IoT device can become either active or inactive according {to} a certain probability. If $D_i$ is active during one time slot, it first samples fresh information and generates a status update packet at the beginning of the time slot, which is known as the ``generate-at-will" model in the literature. $D_i$ then sends the generated status update to the AP. Otherwise, if $D_i$ chooses to be inactive, it stays idle during the said time slot. Moreover, it is assumed that collisions happen if more than one devices become active during the same time slot. We consider the interference-limited regime such that the transmission of status updates fails only when there is a collision. If no collision occurs, the status update of the IoT device is correctly decoded by the AP.

In the following two subsections, we first formally define the average AoI, and then describe the proposed age-dependent random access (ADRA) protocol.
\subsection{{ Average AoI}}\label{AoIdefinition}
\begin{figure}
\centering \scalebox{0.35}{\includegraphics{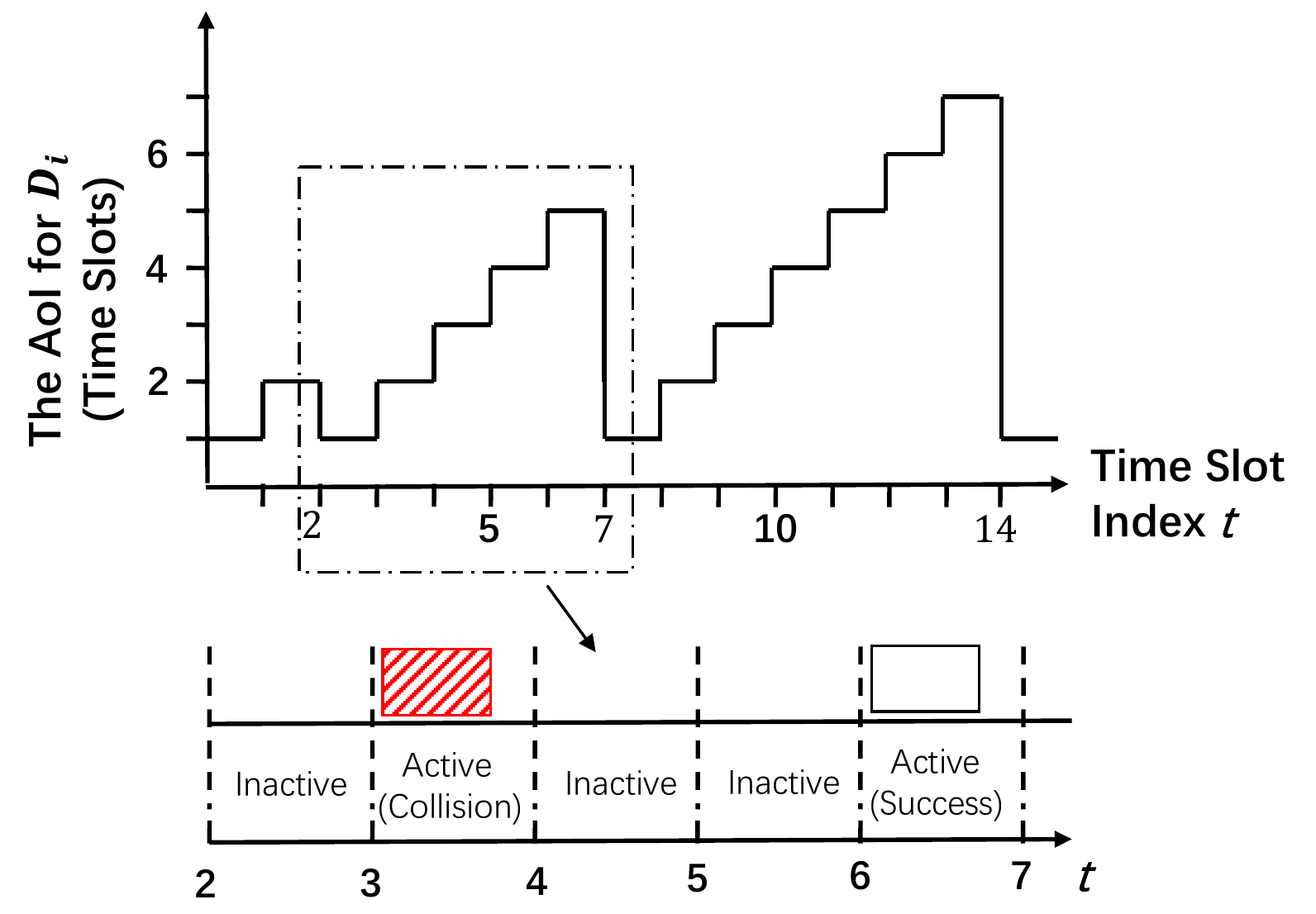}}
\caption{One possible evolution of the instantaneous AoI for the device $D_i$ versus time slot index. }\label{fig:AoI}
	\vspace{-1em}
\end{figure}
Denote by $t=1,2,3,\cdots$ the index of time slots and denote by $\Delta_i\left(t\right)$, $i\in\left\{1,2,\cdots,N\right\}$, the instantaneous AoI of the $i$-th IoT device in time slot $t$. We use $I_i\left(t\right)$ to denote the indicator of the active or inactive status for the device $D_i$ in time slot $t$. Particularly, $I_i\left(t\right)=1$ means that $D_i$ is active during time slot $t$, and $I_i\left(t\right)=0$ otherwise. Based on the definition of the AoI, the instantaneous AoI of $D_i$ drops to one when $D_i$ is active and all other devices keep inactive, i.e., $D_i$ successfully delivers a status update to the AP. Otherwise, the instantaneous AoI of $D_i$ increases by one for each time slot. Mathematically, the evolution of the instantaneous AoI for the device $D_i$ can be expressed~as
\begin{equation}
\Delta _i \left( {t  + 1} \right) =\left\{
\begin{matrix}
\begin{split}
   &{1}, \quad\text{if}\quad I_i \left( t  \right)=1, I_j \left( t  \right) = 0, \forall j\ne i\\
   &{\Delta _i \left( {t } \right)+1}, \quad\text{otherwise}  \\
\end{split}
\end{matrix}
\right..
\end{equation}
To ease understanding, we illustrate the evolution of the instantaneous AoI of $D_i$ for 14 consecutive time slots with a starting value of 1 in Fig. \ref{fig:AoI}. {Based on the AoI evolution, the average AoI for each IoT device is defined as}
\begin{equation}\label{AoIexpression}
\bar \Delta _{i}  = \mathop {\lim }\limits_{T \to \infty } {1 \over T}\sum\limits_{t  = 1}^T {\Delta _i \left( t  \right)}.
\end{equation}
\subsection{Age-Dependent Random Access}
{We now introduce the proposed ADRA protocol, which is an ALOHA-like stationary random access policy. It is worth pointing out here that \cite{S.Kaul-Distributed-Centralized-MAC, R.Talak-Distributed-MAC,chen2019age} also studied ALOHA-like stationary random access policies. However, they all considered age-independent random access (AIRA) protocols. More specifically, devices will access the channel with the same probability no matter whether their instantaneous AoI values are low or high. \emph{Our consideration of the age-dependent policy is motivated by the intuition that those devices with relatively smaller AoI should access the channel with a lower probability such that other devices with larger AoI can achieve a higher success probability to update their statuses by encountering less collisions}. In such a way, all devices co-exist in a more harmonic way to reduce the network-wide AoI together.}

{In our ADRA protocol, all IoT devices maintain a fixed age-dependent channel access probability (CAP) vector $\mathbf{p}=\left\{p_1,p_2,p_3,\cdots, \right\}$, where $p_l$ denotes the active probability when the instantaneous AoI is equal to $l$. As the first attempt to design and evaluate the ADRA policy, in this paper we consider the preliminary case that the elements in the {CAP vector} $\mathbf{p}$ can only equal to either $0$ or $p$. Specifically, if the instantaneous AoI is no less than a threshold $\delta$, the IoT device becomes activate with a fixed probability of $p$. Otherwise, the IoT device will stay inactive with probability 1. Hereafter, we refer to this simplified protocol as the \emph{threshold-based ADRA}. The general case of the CAP will be left as a future work for this paper.} Note that all devices will statistically have the same average AoI due the symmetric structure of the considered model, we thus can drop the subscript of device index $i$ in our subsequent analysis.

\section{Average AoI Analysis}
In this section, we analyze the average AoI of each IoT device for the proposed threshold-based ADRA policy. Here, we clarify that the approaches used in \cite{S.Kaul-Distributed-Centralized-MAC, R.Talak-Distributed-MAC,chen2019age} for analyzing AIRA policies are no longer suitable for our case. Specifically, in age-independent policies, the CAP for each IoT device transmits with a fixed value $p^\prime$ independent of its instantaneous AoI. That is, $p_l=p^\prime, \forall l$. Recall that the IoT device can successfully transmit a status update only when all the  other $N-1$ IoT devices are inactive at the same time. Let $q$ denote the successful status update probability when an IoT device becomes active. We thus have $q=\left(1-p^\prime\right)^{N-1}$. With reference to \cite[Eq. (9)]{S.Kaul-Distributed-Centralized-MAC}, the average AoI of each IoT device for AIRA policies can be readily given~by\footnote{ Note that the evolution of AoI follows a sawtooth in \cite{S.Kaul-Distributed-Centralized-MAC} because it grows linearly over time. By contrast, in our considered slotted system, the evolution of AoI updates at the end of each time slot and thus follows a staircase shape. The expression of average AoI is thus slightly different from \cite[Eq. (9)]{S.Kaul-Distributed-Centralized-MAC}.}
{\begin{equation}\label{conventionalAoI}
\bar \Delta^\prime  = {1 \over {p^\prime q }} =  {1 \over {p^\prime \left( {1 - p^\prime } \right)^{N - 1} }}.
\end{equation}}

In the existing AIRA policies, the successful status update probability $q$ of one device is independent of the instantaneous AoI of all other devices. In contract, in our proposed {threshold-based} ADRA protocol, the CAP vector of each IoT device (i.e., $\mathbf{p}$) depends on its instantaneous AoI. Therefore, $q$ depends on the instantaneous AoI of all IoT devices and thus the AoI evolutions of all devices tangle together, which makes the performance analysis of the {threshold-based} ADRA protocol non-trivial. Although the performance of the {threshold-based} ADRA scheme can be analyzed by applying a multi-dimension Markov Chain (MC), the computational complexity of this method grows exponentially as the number of devices $N$ increases, preventing us from further optimizing the proposed scheme.

To tackle the above issue, we adopt a widely-used approximation approach to decouple the tangled evolution of the AoI for all IoT devices. The key assumption that we apply is that the successful probability $q$ for all IoT devices is a constant when they decide to transmit a status update. Note that with this assumption, the value of $q$ is independent of the instantaneous AoI of all other IoT devices, but it is still a function of the age threshold $\delta$ and the CAP $p$. Such an approximation has been used in the literature to analyze the performance of various random access protocols using conventional metrics like throughput and delay, see e.g., \cite{Kwak2005ton,Dai2012twc} and references therein. Good accuracy has been demonstrated especially when the number of nodes $N$ is large. We later show that our simulation results presented in Sec. IV once again confirm the good accuracy of the adopted approximation.

With the approximation described in the previous paragraph, all devices follow an identical state transition process, which can be described by a Discrete-Time Markov Chain (DTMC) depicted in Fig. \ref{fig:markov1} characterized by the parameters $\delta$, $p$, and $q$. To analyze the average AoI of the proposed protocol, we now derive the stationary distribution of the DTMC.

\subsection{DTMC of Each IoT Device}
We consider a DTMC with infinite states and define each state $S_l$, $l \in \left\{1,2,\cdots\right\}$, as the instantaneous AoI being $l$. The transition probability $T_{m,n}$ is defined as the probability of the transition from state $S_m$ to $S_n$, $m,n \in\left\{1,2,\cdots\right\}$. In the proposed threshold-based ADRA, the AoI of each IoT device drops to one when it becomes active and all the other IoT devices stay inactive. Otherwise, the AoI increases by one. Recall that each IoT device can become active with a fixed probability $p$, when the instantaneous AoI is not less than the threshold $\delta$. Based on this fact, all the non-null transition probabilities of the DTMC can be summarized as
\begin{equation}
\left\{ \begin{matrix}
\begin{split}
   &{T_{l,l + 1}  = 1 }, \quad l \in \left\{ {1,2,\cdots,\delta - 1} \right\},\\
   &{T_{l,l+1}  = 1-pq}, \quad l \in \left\{ {\delta,\delta+1,\cdots} \right\},\\
   &{T_{l,1}  = p q},\quad l \in \left\{ {\delta,\delta+1,\cdots} \right\}.
\end{split}
\end{matrix}
\right..
\end{equation}
\begin{figure}
\centering \scalebox{0.35}{\includegraphics{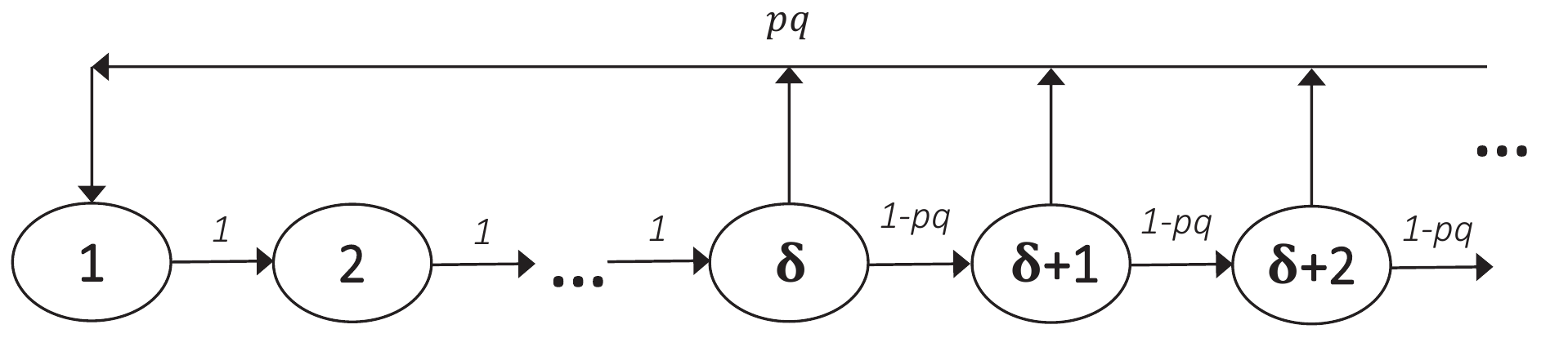}}
\caption{Discrete-Time Markov Chain model of each device in threshold-based ADRA. }\label{fig:markov1}
	\vspace{-1em}
\end{figure}
It is readily to verify that the considered DTMC is irreducible and thus it admits a stationary distribution. Denote by $\mathbf{\pi}=\left\{\pi_1,\pi_2,\cdots\right\}$ the stationary distribution of the AoI for the IoT device. Each element $\pi_l$, $l=1,2,\cdots$, denotes the stationary probability of the instantaneous AoI being $l$. With the derived transition probabilities, we can deduce that the stationary probabilities $\pi_1=\pi_2=\cdots=\pi_\delta$, and $\pi_l=\pi_\delta\left(1-pq\right)^{l-\delta}$, for $l > \delta$. Since $\sum\limits_{l = 1}^\infty  {\pi _l  = 1}$, the stationary distribution of the DTMC is given~by
\begin{equation}\label{TARApi}
\pi _l=\left\{ \begin{matrix}
\begin{split}
   &{{pq} \over {\delta pq + 1 - pq}}, \quad l \in \left\{ {1,2,\cdots,\delta } \right\}\\
   &{{{pq\left( {1 - pq} \right)^{l - \delta } } \over {\delta pq + 1 - pq}}, l \in \left\{ {\delta+1,\delta+2,\cdots} \right\}}\\
\end{split}
\end{matrix}
\right..
\end{equation}
We note that the only unknown parameter in (\ref{TARApi}) is $q$. In the following, we will derive the relationship between $q$ and the known parameters $\delta$ and $p$. Recall that $q$ is assumed to be independent of the instantaneous AoI. If we use $\eta$ to denote the stationary probability of an IoT device transmitting in a randomly chosen time slot, we then can express $q$ as
\begin{equation}\label{qequation}
q= (1- \eta)^{N-1},
\end{equation}
which follows due to the fact that the status update of one device can be successfully delivered only when all other devices are inactive in the said time slot. Note that here we have made a simplifying decoupling assumption the states of the IoT devices are independent of each other.

Moreover, based on the stationary probabilities given in (\ref{TARApi}), and the fact that each IoT device transmits with a fixed probability $p$ when the AoI is no smaller than $\delta$, we can express $\eta$ by the following equation:
\begin{equation}\label{activeprobability}
\begin{split}
\eta  &= \sum\limits_{l = \delta}^\infty  {\pi_l p }={{\pi _1 } \over q} = {p \over {\delta pq + 1 - pq}}.
\end{split}
\end{equation}

Jointly considering (\ref{qequation}) and (\ref{activeprobability}), we now have the following equation for the successful probability $q$:
\begin{equation}\label{solveq}
{1 \over f\left(q\right)}+q^{1\over{N-1}}-1=0,
\end{equation}
where
\begin{equation}
f\left(q\right) = \delta q + {1 \over p}-q.
\end{equation}
For the conventional AIRA policy with $\delta$ being $1$, it can be solved from (\ref{solveq}) that $q=\left(1-p\right)^{N-1}$. This observation coincides with the conventional analysis given above (\ref{conventionalAoI}). In the proposed threshold-based ADRA, we need to solve the successful probability $q$ from the equation given in (\ref{solveq}).

It is known that for the case of $\delta=1$, the optimal $p$ to minimize average AoI is $1/N$ \cite{S.Kaul-Distributed-Centralized-MAC}. For the general case in which $\delta$ is not limited to 1, the optimal $p$ should be more than $1/N$. The intuition is as follows. When $\delta > 1$, some of the devices will be in state $i < \delta$, in which case they will not transmit. Thus, effectively, in any given time slot, the number of devices who will transmit with probability $p$ is likely to be less than $N$. The effective number of devices that compete is smaller than $N$. To take that into account, in our subsequent analysis and simulation work, we assume $p \le 2/N$ so as to explore for $ p > 1/N$. We will show that with $\delta > 1$, we can find $p>1/N$ such that the average AoI is smaller than the average AoI of the case where $\delta = 1, p=1/N$.

To obtain the value of $q$ when $p \le 2/N$, we now define $g\left(q\right)={1 \over f\left(q\right)}+q^{1\over{N-1}}-1$, and present the following three lemmas. Specifically, Lemma 1 characterizes the range of $q$ for a feasible solution. Lemma 2 describes the monotonicity of the function $g\left(q\right)$. Finally, Lemma 3 indicates the uniqueness of the solution to (\ref{solveq}).
\begin{lemma}\label{lemma1}
The successful probability $q$ satisfies the inequality $\left( {{{N - 2} \over N}} \right)^{N - 1}\le q \le 1$, for $0<p \le {2\over N}$.
\end{lemma}
\begin{proof}
When the channel access probability of each IoT device $p \le {2\over N}$, we can deduce that the active probability of each IoT device in a randomly chosen time slot $\eta \le {2 \over N}$. Based on (\ref{qequation}), we have the lower bound $q \ge \left( {{{N - 2} \over N}} \right)^{N - 1}$. On the other hand, the active probability in a randomly chosen time slot $\eta \ge 0$ and the successful probability $q \le 1$. Note that the left-hand side and right-hand side equalities for the bounds only hold when $p = {2\over N}$, $\delta =1$, and $p = 0$, respectively. This completes the proof.
\end{proof}
\begin{lemma}\label{lemma2}
The function $g\left(q\right)$ is a monotonically increasing function for $\left( {{{N - 2} \over N}} \right)^{N - 1}\le q \le 1$, $N \ge 3$ and $p \le {2\over N}$.
\end{lemma}
\begin{proof}
We now prove that function $g\left(q\right)$ is a monotonically increasing function of $q$ for $\left( {{{N - 2} \over N}} \right)^{N - 1}\le q \le 1$. The first-order derivative of $g\left(q\right)$ with respect to $q$ can be calculated~as
\begin{equation}\label{APP1_ineuqality1}
{{dg\left( q \right)} \over {dq}} =  - {{df\left( q \right)/dq} \over {\left[ {f\left( q \right)} \right]^2 }} + {1 \over {\left( {N - 1} \right)q^{{{N - 2} \over {N - 1}}} }},
\end{equation}
where ${{df\left( q \right)} \over {dq}} = \delta-1$. It is not straightforward to prove ${{dg\left( q \right)} \over {dq}}>0$ directly due to the complicated structure of ${{dg\left( q \right)} \over {dq}}$ given in (\ref{APP1_ineuqality1}). Therefore, we now derive an upper bound of ${{df\left( q \right)} \over {dq}}$ to further simplify (\ref{APP1_ineuqality1}), and then prove ${{dg\left( q \right)} \over {dq}}>0$.

With $p \le {2\over N}$, we have
\begin{equation}\label{upperbounddfdq}
{{df\left( q \right)} \over {dq}}\le {{f\left( q \right)} \over q} - {N \over {2q}}.
\end{equation}
Based on the results derived in (\ref{APP1_ineuqality1}) and (\ref{upperbounddfdq}), ${{dg\left( q \right)} \over {dq}}>0$ can be further simplified to
\begin{equation}\label{APP1_inequality3}
{\left[ {f\left( q \right) - {{N - 1} \over {2{q^{{1 \over {N - 1}}}}}}} \right]^2} - {\left( {{{N - 1} \over {2{q^{{1 \over {N - 1}}}}}}} \right)^2} + {{\left(N - 1\right)N} \over {2{q^{{1 \over {N - 1}}}}}} >0.
\end{equation}
A sufficient condition for inequality (\ref{APP1_inequality3}) is  ${{\left(N - 1\right)N} \over {2{q^{{1 \over {N - 1}}}}}} >{\left( {{{N - 1} \over {2{q^{{1 \over {N - 1}}}}}}} \right)^2} $, leading to the inequality
\begin{equation}
q^{{1 \over {N - 1}}}  \ge {{(N - 1)} / {2/N}}.
\end{equation}
Recall that $\left( {{{N - 2} \over N}} \right)^{N - 1} \le q \le 1$, and it can be readily verified that the above inequality holds for $N \ge 3$. This completes the proof.
\end{proof}
\begin{lemma}\label{lemma3}
The equation $g\left(q\right)=0$ admits a unique solution of $q$ for $\left( {{{N - 2} \over N}} \right)^{N - 1}\le q \le 1$, $N \ge 3$ and $p \le {2\over N}$.
\end{lemma}
\begin{proof}
First of all, because $\delta \ge 1$ and $ p \le {2 \over N}$, we have
\begin{equation}
f\left( q \right) \ge \left(\delta -1 \right)q + {2 \over N} = {2 \over N}.
\end{equation}
Therefore, we can attain
\begin{equation}\label{gleft}
g\left(\left( {{{N - 2} \over N}} \right)^{N - 1}\right)\le {N \over 2}+{{{N - 2} \over N}}-1=0.
\end{equation}
It can also be verified that $g\left(1\right)={1 \over {f\left( 1 \right)}}\ge 0$. Note that the equality in (\ref{gleft}) holds only when $p={2 \over N}$ and $\delta = 1$, and the equality for $g\left(1\right)={1 \over {f\left( 1 \right)}}= 0$ holds only when $p=0$. In these cases, the solution of $q$ is given by $q= \left( {{{N - 2} \over N}} \right)^{N - 1}$, and $q=1$, respectively. 

For other general cases, we have $g\left(\left( {{{N - 2} \over N}} \right)^{N - 1}\right)<0$ and $g\left(1\right)>0$. With the monotonicity proved in Lemma 2, we can deduce that there must exist a unique solution $\left({{{N - 2} \over N}} \right)^{N - 1}<q^*<1$ such that $g\left(q^*\right)=0$. This completes the proof.
\end{proof}
Based on the above Lemmas 1-3, we can obtain the successful probability $q$ by solving (\ref{solveq}) using numerical methods like the bisection method. With the solution of $q$, we can then obtain the stationary distribution of the AoI given in (\ref{TARApi}) for the proposed threshold-based ADRA protocol.
\begin{figure*}
\centering
 \subfigure[When $p = 2/N$, $\delta = N$]
  {\scalebox{0.45}{\includegraphics {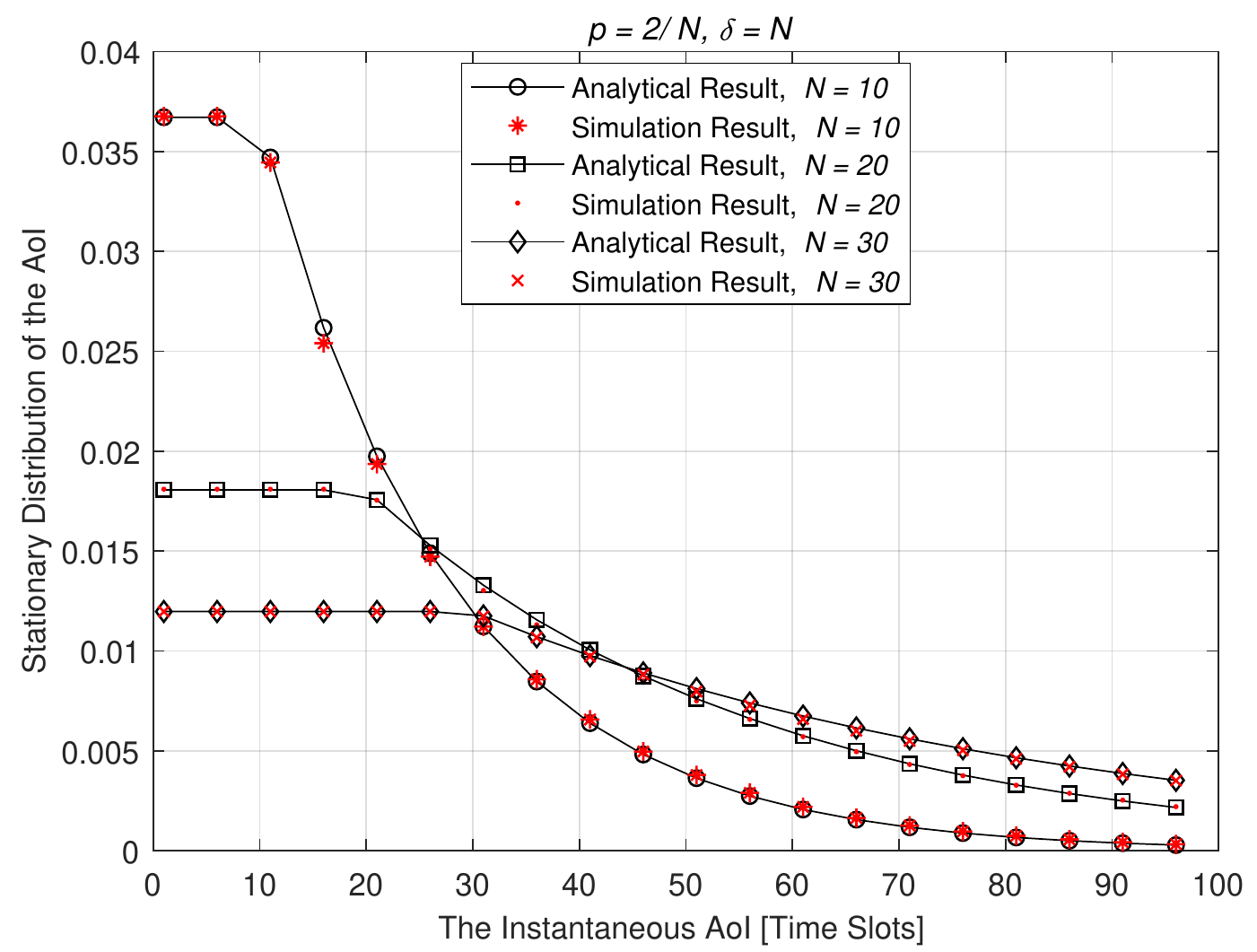}
  \label{fig:distribution_a}}}
\hfil
 \subfigure[When $p = 1.5/N$, $\delta = N/2$]
  {\scalebox{0.45}{\includegraphics {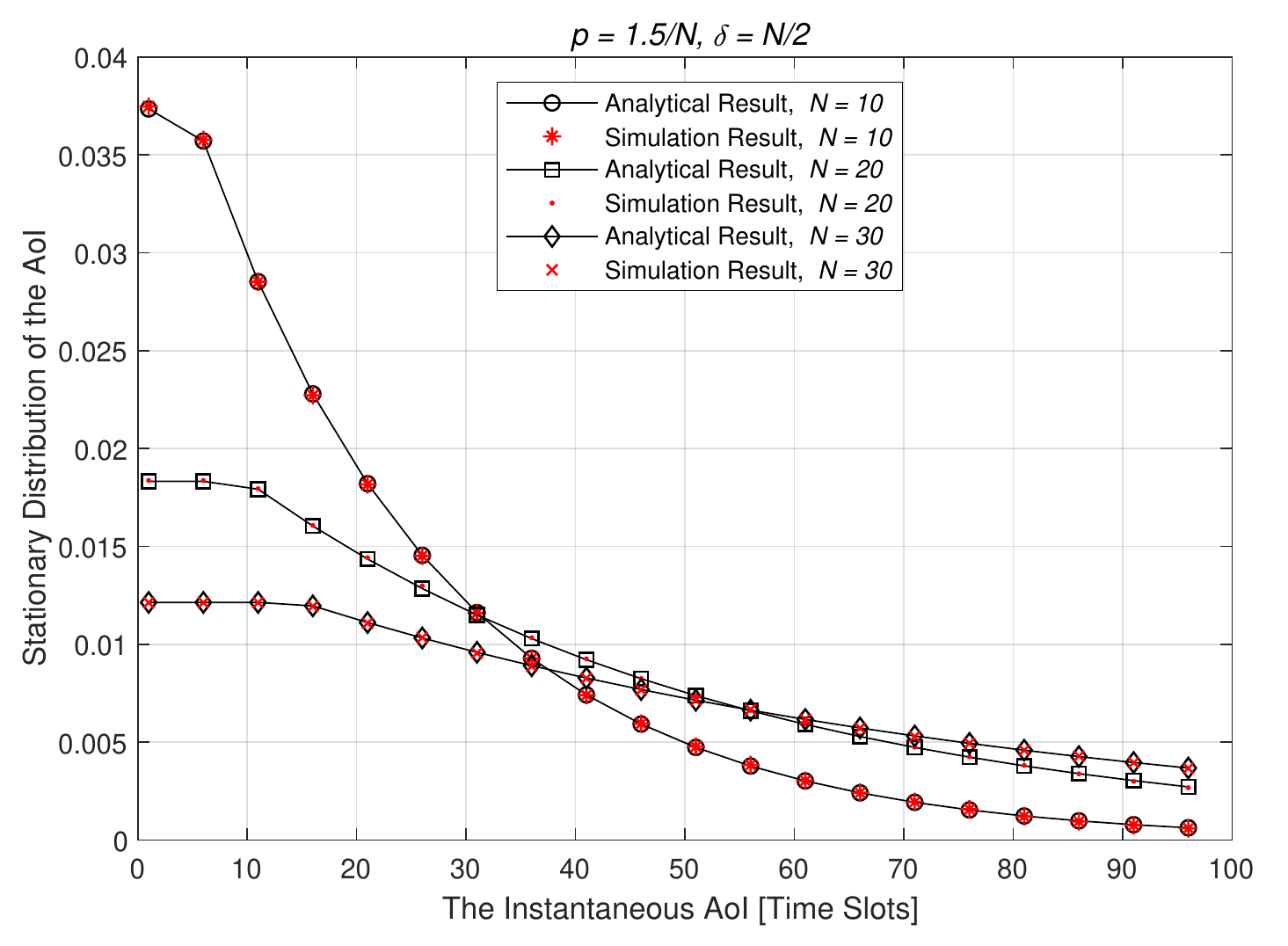}
\label{fig:distribution_b}}}
\caption{Stationary distribution of the instantaneous AoI for different $N$, $p$ and $\delta$.}
\label{fig:distribution}
	\vspace{-1em}
\end{figure*}

\subsection{Analysis of the Average AoI}
The average AoI for each IoT device applying threshold-based ADRA policy can be derived as
\begin{equation}\label{averageAoI}
\bar \Delta=\sum\limits_{l = 1}^\infty  {l\pi _l }={\delta  \over 2} + {1 \over {pq}} - {\delta  \over {2\left( {\delta pq + 1 - pq} \right)}}.
\end{equation}
\begin{remark}\label{remark2}
We can clearly see from (\ref{averageAoI}) that the average AoI of the proposed threshold-based ADRA protocol can be reduced by increasing the active probability $p$ or decreasing the age threshold $\delta$. However, increasing $p$ and decreasing $\delta$ will also lead to frequent collisions which in turn reduce the successful probability $q$, thereby resulting in an increase of the average AoI. Therefore, we can deduce that there exists optimal values of $p$ and $\delta$, which minimize the average AoI evaluated in (\ref{averageAoI}). Due to the complicated structure of the derived analytical expression, it is hard for us to characterize closed-form solutions to the optimal $p$ and $\delta$. Fortunately, the optimal values of $p$ and $\delta$ can easily be obtained via a two-dimension search based on the derived analytical expression.

Moreover, for the special case with $\delta = 1$ (i.e., conventional AIRA policy), we can obtain that the average AoI can be simplified as
\begin{equation}
\bar \Delta'={1 \over {p\left( {1 - p} \right)^{N - 1} }},
\end{equation}
which coincides with the expression given in (\ref{conventionalAoI}).
\end{remark}

\section{Simulation Results and Discussions}

{In this section, we first present simulation results to validate the theoretical analysis conducted in Section III. We then compare the optimal performance of the proposed threshold-based ADRA with that of the existing AIRA in terms of the average AoI.

In Fig. \ref{fig:distribution}, we illustrate the stationary distribution of the instantaneous AoI for different system setups. We can first observe from Fig. \ref{fig:distribution} that the simulation results are close to their analytical counterparts for all the simulated cases. More importantly, the simulation curves and analytical curves coincide each other well when $N$ is relatively large and $p$ is relatively small. This is understandable since we adopt the approximation of a fixed successful probability for all IoT devices in this paper and such approximation has a good accuracy when the number of IoT devices is large and the channel access probability of each device is relatively small. Because the analytical results agree well with the simulation results, we will only plot the analytical curves in the remaining two figures in this section.

We next plot the average AoI of the proposed threshold-based ADRA versus the threshold $\delta$ for different system setups in Fig. \ref{fig:delta}. We can observe from the figure that there exists an optimal value of $\delta$ which minimizes the average AoI of the proposed scheme in all the simulated cases, which validates the deduction in Remark 1. Recall that increasing $\delta$ will, on the one hand, decrease the active probability of each IoT device, but on the other hand, reduce the collision probability of the network. Thus, the optimal value of the average AoI for the proposed threshold-based ADRA can be tuned by finding the optimal $\delta$. Furthermore, the optimal $\delta$ decreases as $N,p$ reduce. The rationale is that each IoT device should transmit more aggressively, i.e., using a lower $\delta$, when the number of nodes are small and the active probability is low. Similar results can also be found for the active probability $p$.

Lastly, we compare the performance of the proposed threshold-based ADRA with the existing AIRA in terms of the average AoI in Fig. \ref{fig:comparison}. The optimal average AoI of the ADRA scheme is obtained by finding the optimal values of $\delta$ and $p$ through an exhaustive search. The optimal active probability for the existing AIRA is set to $1 \over N$ according to \cite{S.Kaul-Distributed-Centralized-MAC}. From Fig. \ref{fig:comparison}, we can clearly see that our proposed threshold-based ADRA can boost the average AoI performance significantly compared with the AIRA. This is because all the IoT devices adopting threshold-based ADRA co-exist in a more harmonious way by providing the nodes with higher AoI more opportunities to transmit without collisions.} More intuitively, we divide the devices into two groups: those with instantiates age smaller than $\delta$ and those larger than $\delta$. There is no contention among devices with states from $1$ to $\delta - 1$. If we can put many devices in this group, with very few devices left in the other group, then the contention among the devices in the second group will be small. {On the other hand, our proposed algorithm also outperforms Algorithm 2 designed in \cite{chen2019age} as the latter focused on the scenario with the number of devices approaching infinity.} Furthermore, we can observe from Fig. \ref{fig:comparison} that the performance gaps between the optimized ADRA protocol and the benchmarking schemes are enlarged as the number of end-devices increases. Note that although we limit our search space to $p<2/N$, the result is already better than the case of $p=1/N,\;\delta=1$. An outstanding question is whether we could find even better solutions if we extend our search space to beyond $p<2/N$.

\begin{figure}
\centering \scalebox{0.45}{\includegraphics{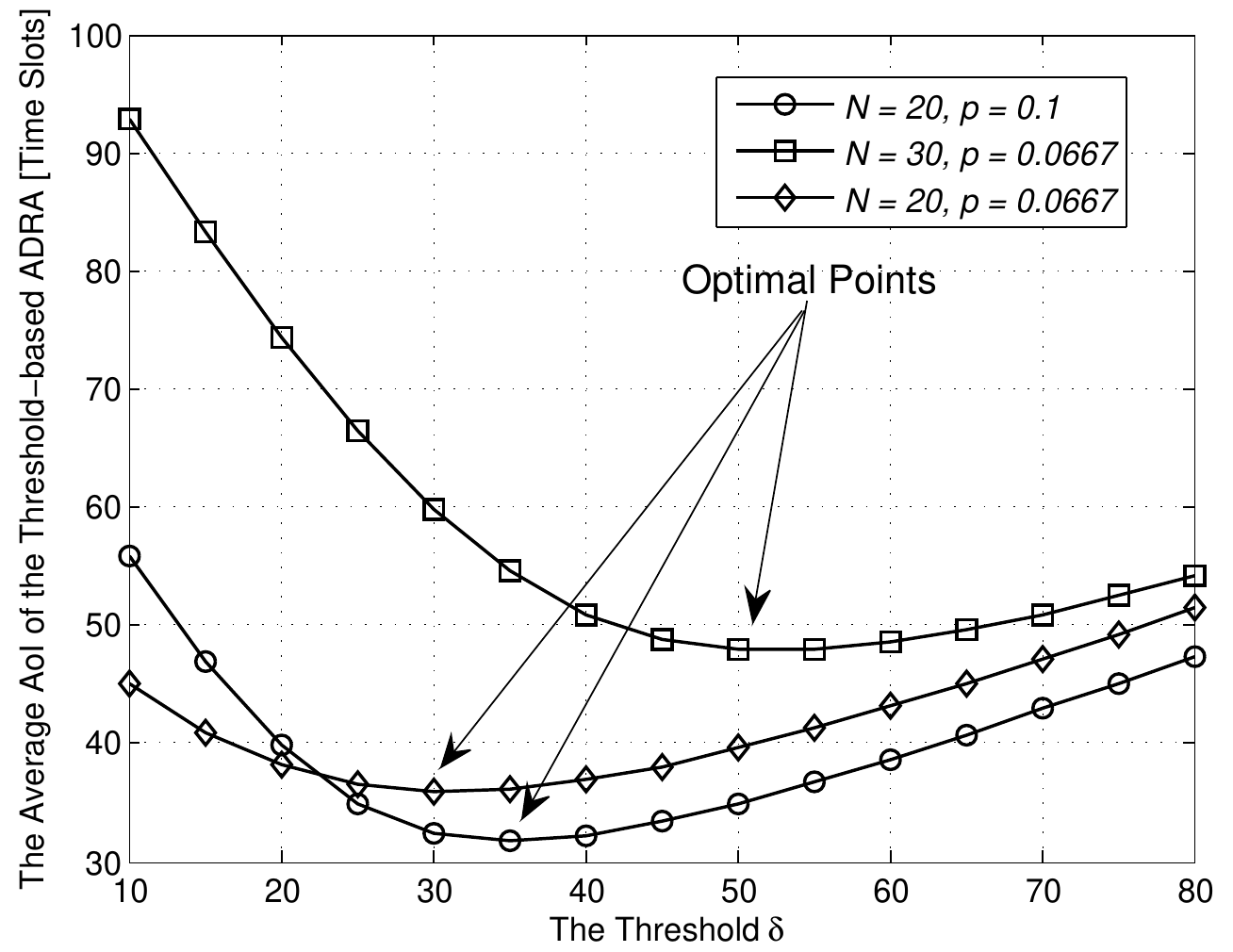}}
\caption{The average AoI of the threshold-based ADRA versus threshold $\delta$ for different values of $N$ and $p$. }\label{fig:delta}
	\vspace{-1em}
\end{figure}

\begin{figure}
\centering \scalebox{0.45}{\includegraphics{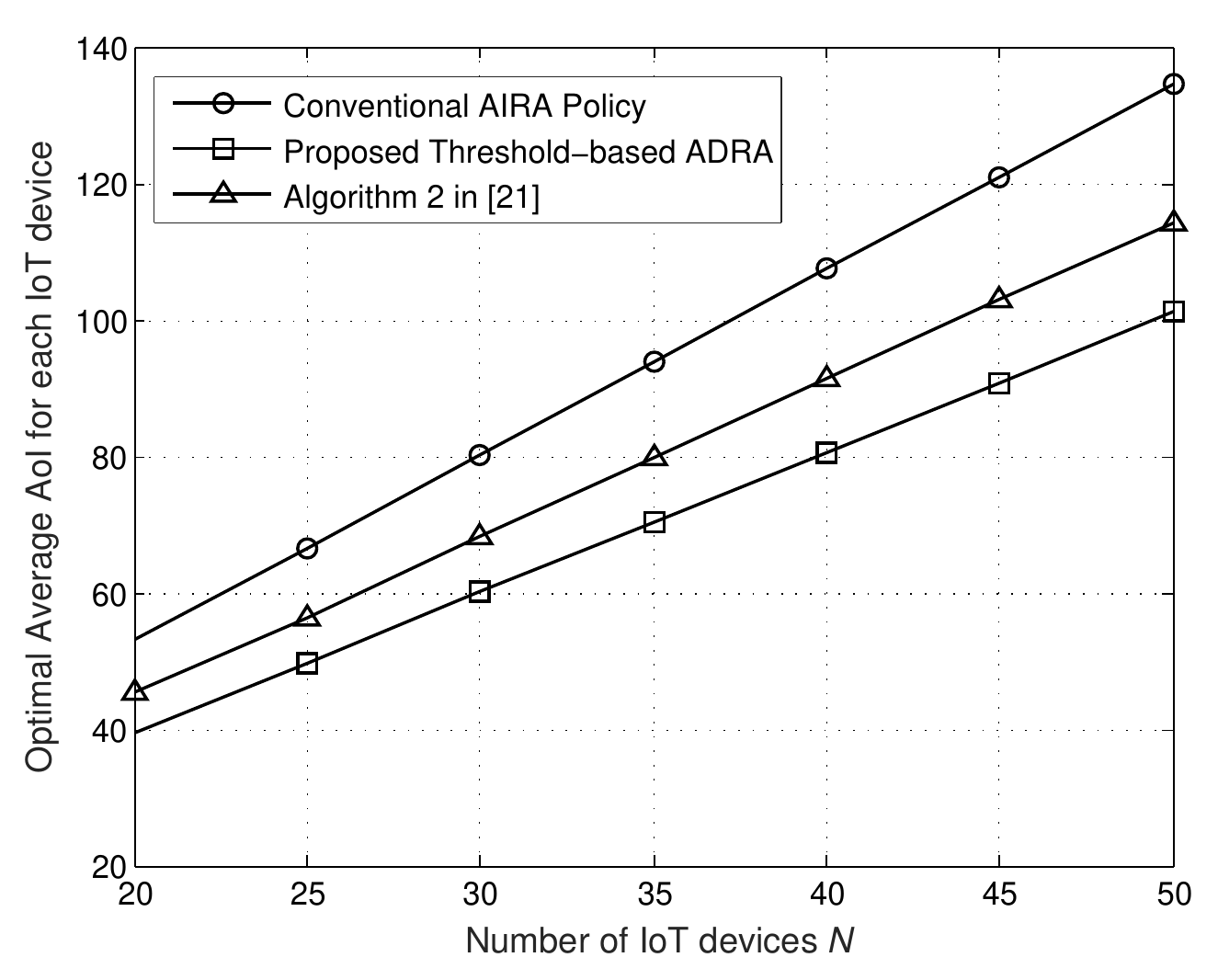}}
\caption{The average AoI versus the number of IoT devices $N$ for the proposed threshold-based ADRA, the existing AIRA schemes and Algorithm 2 in \cite{chen2019age}. }\label{fig:comparison}
	\vspace{-1em}
\end{figure}


%
{
\section{Conclusions}
In this paper, we proposed a threshold-based age-dependent random access (ADRA) scheme for massive IoT networks. In ADRA, each IoT device can only become active and transmit its latest status when its instantaneous age is no less than a predefined threshold. To analyze the average performance of the ADRA, we adopted an approximation to decouple the tangled AoI evolutions of all IoT devices. Specifically, we approximately modelled the AoI evolution for each IoT device by a Discrete-Time Markov chain with a fixed successful probability that is irrelevant to the instantaneous AoI of all IoT devices. Simulation results verified the tightness of our approximation and the correctness of our theoretical analysis, and showed that the proposed threshold-based ADRA scheme outperform the state-of-the-art age-oriented random access scheme in all simulated cases.}

\ifCLASSOPTIONcaptionsoff
  \newpage
\fi

\bibliographystyle{IEEEtran}
\bibliography{References}

\begin{thebibliography}{10}
\providecommand{\url}[1]{#1}
\csname url@samestyle\endcsname
\providecommand{\newblock}{\relax}
\providecommand{\bibinfo}[2]{#2}
\providecommand{\BIBentrySTDinterwordspacing}{\spaceskip=0pt\relax}
\providecommand{\BIBentryALTinterwordstretchfactor}{4}
\providecommand{\BIBentryALTinterwordspacing}{\spaceskip=\fontdimen2\font plus
\BIBentryALTinterwordstretchfactor\fontdimen3\font minus
  \fontdimen4\font\relax}
\providecommand{\BIBforeignlanguage}[2]{{%
\expandafter\ifx\csname l@#1\endcsname\relax
\typeout{** WARNING: IEEEtran.bst: No hyphenation pattern has been}%
\typeout{** loaded for the language `#1'. Using the pattern for}%
\typeout{** the default language instead.}%
\else
\language=\csname l@#1\endcsname
\fi
#2}}
\providecommand{\BIBdecl}{\relax}
\BIBdecl

\bibitem{Lin2017IoTJ}
J.~{Lin}, W.~{Yu}, N.~{Zhang}, X.~{Yang}, H.~{Zhang}, and W.~{Zhao}, ``A survey
  on internet of things: Architecture, enabling technologies, security and
  privacy, and applications,'' \emph{IEEE Internet of Things Journal}, vol.~4,
  no.~5, pp. 1125--1142, Oct 2017.

\bibitem{Ericsson2016}
Ericsson, ``Cellular networks for massive {IoT},'' \emph{Tech. Rep.}, Jan 2016.

\bibitem{McKinsey2015}
M.~G. Institute, ``The {I}nternet of {T}hings: Mapping the value beyond the
  hype,'' June 2015.

\bibitem{Shreedhar2018Acp}
T.~Shreedhar, S.~K. Kaul, and R.~D. Yates, ``Acp: Age control protocol for
  minimizing age of information over the internet,'' ser. MobiCom'18, New York,
  NY, USA, 2018, p. 699–701.

\bibitem{Sun2017tit}
Y.~{Sun}, E.~{Uysal-Biyikoglu}, R.~D. {Yates}, C.~E. {Koksal}, and N.~B.
  {Shroff}, ``Update or wait: How to keep your data fresh,'' \emph{IEEE
  Transactions on Information Theory}, vol.~63, no.~11, pp. 7492--7508, Nov
  2017.

\bibitem{Kaul2011mini}
S.~{Kaul}, M.~{Gruteser}, V.~{Rai}, and J.~{Kenney}, ``Minimizing age of
  information in vehicular networks,'' in \emph{2011 8th Annual IEEE
  Communications Society Conference on Sensor, Mesh and Ad Hoc Communications
  and Networks}, June 2011, pp. 350--358.

\bibitem{kadota2018scheduling}
I.~Kadota, A.~Sinha, E.~Uysal-Biyikoglu, R.~Singh, and E.~Modiano, ``Scheduling
  policies for minimizing age of information in broadcast wireless networks,''
  \emph{IEEE/ACM Transactions on Networking (TON)}, vol.~26, no.~6, pp.
  2637--2650, 2018.

\bibitem{wang2018skip}
B.~Wang, S.~Feng, and J.~Yang, ``To skip or to switch? minimizing age of
  information under link capacity constraint,'' in \emph{2018 IEEE 19th
  International Workshop on Signal Processing Advances in Wireless
  Communications (SPAWC)}.\hskip 1em plus 0.5em minus 0.4em\relax IEEE, 2018,
  pp. 1--5.

\bibitem{gu2019minimizing}
Y.~Gu, H.~Chen, C.~Zhai, Y.~Li, and B.~Vucetic, ``Minimizing age of information
  in cognitive radio-based iot systems: Underlay or overlay?'' \emph{IEEE
  Internet of Things Journal}, 2019.

\bibitem{maatouk2019minimizing}
A.~Maatouk, M.~Assaad, and A.~Ephremides, ``Minimizing the age of information:
  Noma or oma?'' \emph{arXiv preprint arXiv:1901.03020}, 2019.

\bibitem{gu2019timely}
Y.~Gu, H.~Chen, Y.~Zhou, Y.~Li, and B.~Vucetic, ``Timely status update in
  internet of things monitoring systems: An age-energy tradeoff,'' \emph{IEEE
  Internet of Things Journal}, 2019.

\bibitem{wang2019minimizing}
Q.~Wang, H.~Chen, Y.~Li, Z.~Pang, and B.~Vucetic, ``Minimizing age of
  information for real-time monitoring in resource-constrained industrial iot
  networks,'' \emph{arXiv preprint arXiv:1912.07186}, 2019.

\bibitem{wang2020minimizing}
Q.~Wang, H.~Chen, Y.~Li, and B.~Vucetic, ``Minimizing age of information via
  hybrid noma/oma,'' \emph{arXiv preprint arXiv:2001.04042}, 2020.

\bibitem{wang2019minimizing2}
Q.~Wang, H.~Chen, Y.~Gu, Y.~Li, and B.~Vucetic, ``Minimizing the age of
  information of cognitive radio-based iot systems under a collision
  constraint,'' \emph{arXiv preprint arXiv:2001.02482}, 2020.

\bibitem{li2020ageoriented}
B.~Li, H.~Chen, Y.~Zhou, and Y.~Li, ``Age-oriented opportunistic relaying in
  cooperative status update systems with stochastic arrivals,'' \emph{arXiv
  preprint arXiv:2001.04084}, 2020.

\bibitem{S.Kaul-Distributed-Centralized-MAC}
S.~K. {Kaul} and R.~D. {Yates}, ``Status updates over unreliable multiaccess
  channels,'' in \emph{Proc. 2017 IEEE International Symposium on Information
  Theory (ISIT)}, June 2017, pp. 331--335.

\bibitem{R.Talak-Distributed-MAC}
R.~{Talak}, S.~{Karaman}, and E.~{Modiano}, ``Distributed scheduling algorithms
  for optimizing information freshness in wireless networks,'' in \emph{Proc.
  2018 IEEE 19th International Workshop on Signal Processing Advances in
  Wireless Communications (SPAWC)}, June 2018, pp. 1--5.

\bibitem{Kosta2019age}
A.~{Kosta}, N.~{Pappas}, A.~{Ephremides}, and V.~{Angelakis}, ``Age of
  information performance of multiaccess strategies with packet management,''
  \emph{Journal of Communications and Networks}, vol.~21, no.~3, pp. 244--255,
  June 2019.

\bibitem{jiang2018timely}
Z.~Jiang, B.~Krishnamachari, X.~Zheng, S.~Zhou, and Z.~Niu, ``Timely status
  update in massive iot systems: Decentralized scheduling for wireless
  uplinks,'' 2018.

\bibitem{jiang2018decentralized}
Z.~Jiang, B.~Krishnamachari, S.~Zhou, and Z.~Niu, ``Can decentralized status
  update achieve universally near-optimal age-of-information in wireless
  multiaccess channels?'' 2018.

\bibitem{chen2019age}
X.~Chen, K.~Gatsis, H.~Hassani, and S.~S. Bidokhti, ``Age of information in
  random access channels,'' \emph{arXiv preprint arXiv:1912.01473}, 2019.

\bibitem{Kwak2005ton}
{Byung-Jae Kwak}, {Nah-Oak Song}, and L.~E. {Miller}, ``Performance analysis of
  exponential backoff,'' \emph{IEEE/ACM Transactions on Networking}, vol.~13,
  no.~2, pp. 343--355, April 2005.

\bibitem{Dai2012twc}
L.~{Dai}, ``Stability and delay analysis of buffered aloha networks,''
  \emph{IEEE Transactions on Wireless Communications}, vol.~11, no.~8, pp.
  2707--2719, August 2012.

\end{thebibliography}

\end{document}